\documentclass[11pt]{article}
\usepackage{amsfonts,amsmath,amsthm,array,amssymb}
\newtheorem{theorem}{Theorem}

\usepackage{graphicx}
\usepackage{epsfig}
\usepackage{graphics}
\usepackage{epstopdf}
\usepackage{epspdfconversion}
\usepackage{picinpar,graphicx}
 
\usepackage{multicol}
\usepackage{multirow}
\usepackage[nocompress]{cite}
\usepackage{inputenc}
\usepackage{tikz}
\usepackage{bm}
\usepackage{pgfplots}
\usepackage[format=hang]{caption}
\usepackage{float}

\newcommand{\uto}{U_{t+1}}
\newcommand{\ut}{U_{t}}
\newcommand{\oto}{O_{t+1}}
\newcommand{\ot}{O_{t}}
\newcommand{\sto}{S_{t+1}}
\newcommand{\st}{S_{t}}
\newcommand{\go}{e^{\frac{-k}{N}\ot}}
\newcommand{\gs}{e^{\frac{-k}{N}\st}}
\newcommand{\mmo}{O}
\newcommand{\mmu}{U}
\newcommand{\mms}{S}

\newtheorem{lem}{Lemma}%[section]

\begin{document}
%%------------------------------------------------------------------------%%
%%-----------------------Title Page-----------------------------------%%
%%------------------------------------------------------------------------%%
\title{The Dynamics of Offensive Messages in the World of Social Media: the Control of Cyberbullying on Twitter}
\author{Krystal Blanco$^{1}$, Aida Briceno$^{2}$, Andrea Steele$^{3}$, Javier Tapia$^{4}$, \\ John McKay$^{5}$, Sherry Towers$^{6,7}$, Kamuela E. Yong$^{7,8}$}
\date{}
\maketitle
\begin{center}
\footnotesize $^{1}$ Department of Mathematics, Boston University, Boston, MA\\
\footnotesize $^{2}$ Department of Mathematics and Computing, Columbia College, Columbia, SC\\
\footnotesize $^{3}$ Department of Biology, Medgar Evers College (CUNY), Brooklyn, NY\\
\footnotesize $^{4}$ Department of Mathematics, St. Mary's University, San Antonio, TX\\
\footnotesize $^{5}$ Department of Applied Mathematics for the Life and Social Sciences, Arizona State University, Tempe, AZ\\
\footnotesize $^{6}$ Department of Physics, Purdue University, West Lafayette, IN\\
\footnotesize $^{7}$ Mathematical, Computational \& Modeling Sciences Center, Arizona State University, Tempe, AZ\\
\footnotesize $^{8}$ School of Mathematical \& Statistical Science, Arizona State University, Tempe, AZ\
\end{center}

%%------------------------------------------------------------------------%%
%%-----------------------Abstract--------------------------------------%%
%%------------------------------------------------------------------------%%
\vspace{0.3in}
\begin{abstract} 
The 21st century has redefined the way we communicate, our concept of individual and group privacy, and the dynamics of acceptable behavioral norms. The messaging dynamics on Twitter, an internet social network, has opened new ways/modes of spreading information. As a result cyberbullying or in general, the spread of offensive messages, is a prevalent problem. The aim of this report is to identify and evaluate conditions that would dampen the role of cyberbullying dynamics on Twitter. 
We present a discrete-time non-linear compartmental model to explore how the introduction of a Quarantine class may help to hinder the spread of offensive messages. 
We based the parameters of this model on recent Twitter data related to a topic that communities would deem most offensive, and found that for Twitter a level of quarantine can always be achieved that will immediately suppress the spread of offensive messages, and that this level of quarantine is independent of the number of offenders spreading the message. 
We hope that the analysis of this dynamic model will shed some insights into the viability of new models of methods for reducing cyberbullying in public social networks.
\end{abstract}

\newpage

%%------------------------------------------------------------------------%%
%%-----------------------Introduction---------------------------------%%
%%------------------------------------------------------------------------%%

\section{Introduction}
Social media networking websites like Twitter and Facebook have fundamentally changed how people communicate and socialize in the 21st century. With over 200 million users on the popular microblogging site Twitter in 2011 \cite{Qiu2012}, exploring the social impact of using such a service becomes ever more important. Although Twitter has been considered by critics as a site where people tweet (i.e. update statuses) about mundane things such as what they are having for dinner, researchers have recently taken an interest in Twitter to study social-behavioral attitudes \cite{Gruzd2011}, positive and negative influences of popular users \cite {Bae2012}, and even temporal patterns of happiness using ideas from linguistics \cite{Dodds2011}. In this emerging culture of Tweeters (i.e. one who tweets) that connect from places all over the world, there also comes the bigger problem of cyberbullying \cite{Hinduja2008}.

Twitter connects people by allowing them to send short bursts of information, called tweets, consisting of up to 140 characters in length, to other users. What makes Twitter different from other social networking sites is the character limit imposed on tweets, along with its asymmetrical nature: you may follow someone, but they do not have to follow you. A user has three different types of relationships with fellow Tweeters: you can have followers (i.e. friends), which are people who follow you and can see your tweets; you can also have people who you are following; and finally there are bi-directional friends, which are people who are mutually following each other. 

Despite the different relationships that occur on Twitter, when a user tweets to the Twitter community, a response is not necessary. User's friends can see all the tweets sent by the user. To involve a particular person in a tweet or a conversation, common practice is to use `@' followed by a unique identifier address. To join a conversation, one can reply to a tweet, which uses `@'. One can also retweet, which allows a user to tweet the same message to all his/her friends while giving credit to the original tweeter, allowing the message to reach a wider range of Twitter users \cite{Kwak2010}. For simplicity of presentation in this analysis, we use the term ``retweet" to describe both types of message repetition. In their retweets, users can use a hashtag, denoted by `\#,' which declares a tweet as being a part of a larger conversation whose topic is related to the phrase which comes after the hashtag \cite{Zappavigna2011}. We differentiate between keywords and hashtags: hashtags are used to identify the topic of a tweet, while keywords are found in the actual tweet.

Since the advent of social media sites like Twitter, there has been research on cyberbullying \cite{Gorzig2013,Li2006,Xu2012}. Some research into methods for stopping cyberbullying on Twitter can be found in Xu \textit{et al.} \cite{Xu2012}. In their report, the authors used natural processing techniques to recognize traces of bullying tweets \cite{Xu2012}, and Gorzig and Frumkin suggested raising awareness for adolescents, especially with regard to privacy settings \cite{Gorzig2013}. While a user can block another person specifically on Twitter, no current method exists to block offensive or hurtful tweets from the Twitter community. Per the terms of services, Twitter does not currently suspend accounts for cyberbullying (available at {\tt http://twitter.com/tos}). In this study, we employ a concept refered to as ``quarantine" that would enable Twitter to temporarily separate offenders from the Twitter community, therefore protecting users from the spread of offensive messages. 

In this analysis, we extended the deterministic compartmental model developed by Zhao \textit{et al.} for rumor spreading in the new media age \cite{Zhao2013a}. In the model, they built on the pioneering work of Daley and Kendall \cite{Daley1965}, who in 1965 introduced a model that looks at the spread of rumors as a non-standard contagion process. The Daley-Kendall model was the first rumor spreading model developed, it used the terms Ignorants, Spreaders and Stiflers analogous to the disease model classes of Susceptible, Infected and Recovered, respectively. Zhao \textit{et al.} also looked at the Maki-Thompson rumor model \cite{Maki1973}, which showed the spreading of rumors through direct contact between spreaders and others, and the model developed by Nekovee \textit{et al.} which combined the SIR epidemic model and the Maki-Thompson model on complex social networks \cite{Nekovee2007}. By drawing inferences from these models, Zhao \textit{et al.} developed their own model which included a ``Hibernator" class where the members of this class become disinterested in spreading the rumor but could become interested in spreading the same rumor again. We considered instead a model where users are ``quarantined," i.e. users are limited in their ability to spread the rumor through the enforced limitation on their contact with other members of the population. To do this, we developed a discrete-time compartmental model that simulates
%ST: changed simulate to simulates
the dynamics of rumor spreading in social networks, and examined how user's degree of quarantine hinders the spread of offensive messages. We estimated model parameters by comparing the model predictions to patterns observed in Twitter data related to a topic that communities would find offensive.

In the following sections we describe the sources of data used in this analysis and give a description of the discrete-time compartmental model used to simulate
spread of tweets, followed by a presentation and discussion of results.

\tikzstyle{blockblue} = [rectangle, draw, fill=blue!20, 
    text width=5em, text centered, rounded corners, minimum height=4em]
\tikzstyle{blockred} = [rectangle, draw, fill=red!20, 
    text width=5em, text centered, rounded corners, minimum height=4em]
\tikzstyle{blockgreen} = [rectangle, draw, fill=green!20, 
    text width=5em, text centered, rounded corners, minimum height=4em]
\tikzstyle{blockviolet} = [rectangle, draw, fill=violet!20, 
    text width=5em, text centered, rounded corners, minimum height=4em]
%\tikzstyle{line} = [draw, -latex']

%%%%%%%%%%%%%%%%%%%%%%%%%%%%%%%%%%%%%%%%%%%%%%%%%%%%%%%%%%%%%%%%%%%%%%%%%%%%%%%%%%
%%%%%%%%%%%%%%%%%%%%%%%%%%%%%%%%%%%%%%%%%%%%%%%%%%%%%%%%%%%%%%%%%%%%%%%%%%%%%%%%%%
%%%%%%%%%%%%%%%%%%%%%%%%%%%%%%%%%%%%%%%%%%%%%%%%%%%%%%%%%%%%%%%%%%%%%%%%%%%%%%%%%%
\section{Methods and Materials}

\subsection{Data Collection} 
\label{sec:Data}
Twitter uses Application Programming Interface (API) version 1.1 which allows us to ``scrape" public data off the website ({\tt https://twitter.com}). To be able to access the API, we use the R programming language package ``twitteR" ({\tt http://cran.r-project.org/web/packages/twitteR/}). This package acquires the information about a tweet such as the time created, the screen name of the person who tweeted or retweeted, and their follower and friend counts. In order to decide which keywords to use, we conducted a search on the Twitter website and searched for trending hashtags which contained offensive language. We noticed that anti-gay sentiment is prevalent on Twitter so we chose words that expressed disgust towards the gay community to use as an example. There is no doubt that these tweets are offensive to most people. We used the ``searchTwitter" function to search within the days of July 6-10, 2013 for offensive tweets or retweets containing the keywords ``disgusting" and ``gay" yielding a total of 884 tweets. The data showed that there were 100 tweets from this sample of 478 original tweets that had been retweeted at least once. This is in qualitative concordance with Kwak \textit{et al.} where they found the fraction of not retweeting to be 79\% \cite{Kwak2010}, and Grabowicz \textit{et al.} who found that 85\% of tweets were not retweeted \cite{Grabowicz2012}. The average duplication of a tweet in our sample was 1.85 times. In addition, we estimated the probability of a message being retweeted by using $(1-\alpha)=$ (number of retweets)/(number of friends of offensive tweeters), yielding $(1-\alpha)$ = $478 / 275960$. Thus $\alpha$ $\sim$ 0.999, which will be included as a parameter in our model below.

Rodrigues \textit{et al.} found that 99\% of Twitter users have fewer than 20 followers \cite{Rodrigues2011}. 
Furthermore, the degree of social networks has been found to follow a power law functional relationship, with exponent $2 \le \gamma \le 3$ \cite{Kwak2010}.  The probability of having $m$ Twitter followers is thus $P(m) = m^{-\gamma}$ \cite{Kwak2010}.  With the power law relationship, most users have few followers and a very few users have many followers.  The average degree of the network, $k$, is $E[m]$, which has a value of 1.5 to 10 for $\gamma
\in [2,3]$.  For our model analysis below, we assume $k$ = 10. 

%%%%%%%%%%%%%%%%%%%%%%%%%%%%%%%%%%%%%%%%%%%%%%%%%%%%%%%%%%%%%%%%%%%%%%%%%%%%%%%%%%
%%%%%%%%%%%%%%%%%%%%%%%%%%%%%%%%%%%%%%%%%%%%%%%%%%%%%%%%%%%%%%%%%%%%%%%%%%%%%%%%%%
%%%%%%%%%%%%%%%%%%%%%%%%%%%%%%%%%%%%%%%%%%%%%%%%%%%%%%%%%%%%%%%%%%%%%%%%%%%%%%%%%%
\subsection{Model}
We developed a discrete-time model to simulate the dynamics of the spread of messages on Twitter, because tweets occur at discrete times, whereas a continuous-time model includes a variable window of times between the tweets. In the following subsections, we describe a basic model of Twitter message spreading dynamics that does not include a Quarantine class. We determined an expressions for the threshold value of this model, then we extended the model to include a Quarantined class and discuss the implications that the addition of this class has for the threshold values of the system.   

\subsubsection{Basic Twitter Model Without Quarantine Class} \label{Section:Basic}
A simple discrete-time rumor model for describing the spread of a message in a social network is:
\begin{align}
U_{t+1} & = U_{t} G_{1,t} \nonumber \\
O_{t+1} & = (1 - \alpha) (1 - G_{1,t})  U_{t} + G_{2,t} O_{t} \nonumber  \\
S_{t+1} & = S_{t} + \alpha (1 - G_{1,t}) U_{t} + (1 - G_{2,t}) O_{t}, \label{eqn:basic}
\end{align}
which was proposed by Zhao \textit{et al.} \cite{Zhao2013a}, where $U_{t}$, $O_{t}$, $S_{t}$ are the Uninformed, Offender, and Stifler classes at time $t$, respectively, where the constant population size is $N = U_{t} + O_{t} + S_{t}$, and $\alpha$ is the probability that an uninformed user becomes a stifler after seeing an offensive tweet (which we estimated from Twitter data in the previous section).  The parameter $(1-G_{1,t})$ is the probability of seeing an offensive tweet, and $(1-G_{2,t})$ is the probability of an offender becoming a stifler. To ensure that these two probabilities lie between 0 and 1, we use the form \cite{Hernandez-Ceron2013}:
\begin{align}
G_{1,t} & = e^{-\frac{k}{N} (O_{t})} \nonumber \\
G_{2,t} & = e^{-\frac{k}{N} (S_{t})} \nonumber, 
\end{align}
where $k$ is the degree of the network, which is estimated from studies of Twitter data in the published literature, as described in the previous section.  
%The population within each of the classes at any time
%step is constrained to be at
%least zero, and at most $N$.
 The parameters of the model are shown in Table \ref{Tab:params} and the compartmental diagram is shown in Figure \ref{Model:BasicComp}.

In Lemma~\ref{lem:gt0} in Appendix A we show that if the initial values of $U$, $O$, and $S$ are non-negative, then all future populations in those classes will be non-negative.  Taken in conjunction with the
fact that the system is bounded such that the populations in each of the classes sums to $N$
(see Lemma~\ref{lem:bound} in Appendix A), we see that 
the population within each class can be at most $N$.
We note that as a result, the population in the uninformed class monotonically 
decreases asymptotically to zero, because $G_{1,t}\le1$ when $k>0$ and
$O_t\ge0$ (for a formal proof, see Lemma~\ref{lem:decU} in Appendix A). Additionally, the population in the Stifler class monotonically increases, because the
last two terms in $S_{t+1} = S_{t} + \alpha (1 - G_{1,t}) U_{t} + (1 - G_{2,t}) O_{t}$ are both
at least zero for all $0\le U_t\le N$, $0\le O_t\le N$, $0\le\alpha\le1$ and $k>0$. Further, when the
population in the Stifler class is $N$ we note that $U_t=0$ and $O_t=0$, and that from
System~\ref{eqn:basic} we see that $S_{t+1}=S_t$.  Thus, $S_t$ asymptotically monotonically increases to $N$.

In Theorem~\ref{thm:thresh} in Appendix A we show that if 
\begin{eqnarray}
  Z = {{k (1-\alpha) U_t}\over{N (1- G_{2,t})}} < 1
\label{eqn:R}
\end{eqnarray}
then $O_{t+1} < O_t$.  Since $U_t$ monotonically decreases to zero, and $(1-e^{-kS_t/N})$ monotonically
increases because $S_t$ monotonically increases, we see that $Z$ montonically decreases to zero.
Thus there will always exist a point at which $Z<1$, and $O_{t+1}<O_{t}$.  We thus conclude
that even without quarantine, the Offender class will eventually begin to die out.
In the next section, we will discuss how the implementation of a Quarantine class ensures
that the Offender class will begin to die out more quickly.

%%%%%%%%%%%%%%%%%%%%%%%%%%%%%%%%%%%%%%%%%%%%%%%%%%%%%%%%%%%%%%%%%%%%%%%%%%%%%%%
%Compartmental Model for the Basic Model                 
%%%%%%%%%%%%%%%%%%%%%%%%%%%%%%%%%%%%%%%%%%%%%%%%%%%%%%%%%%%%%%%%%%%%%%%%%%%%%%%
\tikzstyle{blockblue} = [rectangle, draw, fill=blue!20, 
    text width=5em, text centered, rounded corners, minimum height=4em]
\tikzstyle{blockgreen} = [rectangle, draw, fill=green!20, 
    text width=5em, text centered, rounded corners, minimum height=4em]
\tikzstyle{blockviolet} = [rectangle, draw, fill=violet!20, 
    text width=5em, text centered, rounded corners, minimum height=4em]
%\tikzstyle{line} = [draw, -latex']

\begin{figure}[!ht]% figure placement: here, top, bottom, or page
 \centering
 \includegraphics[height=2.5in]{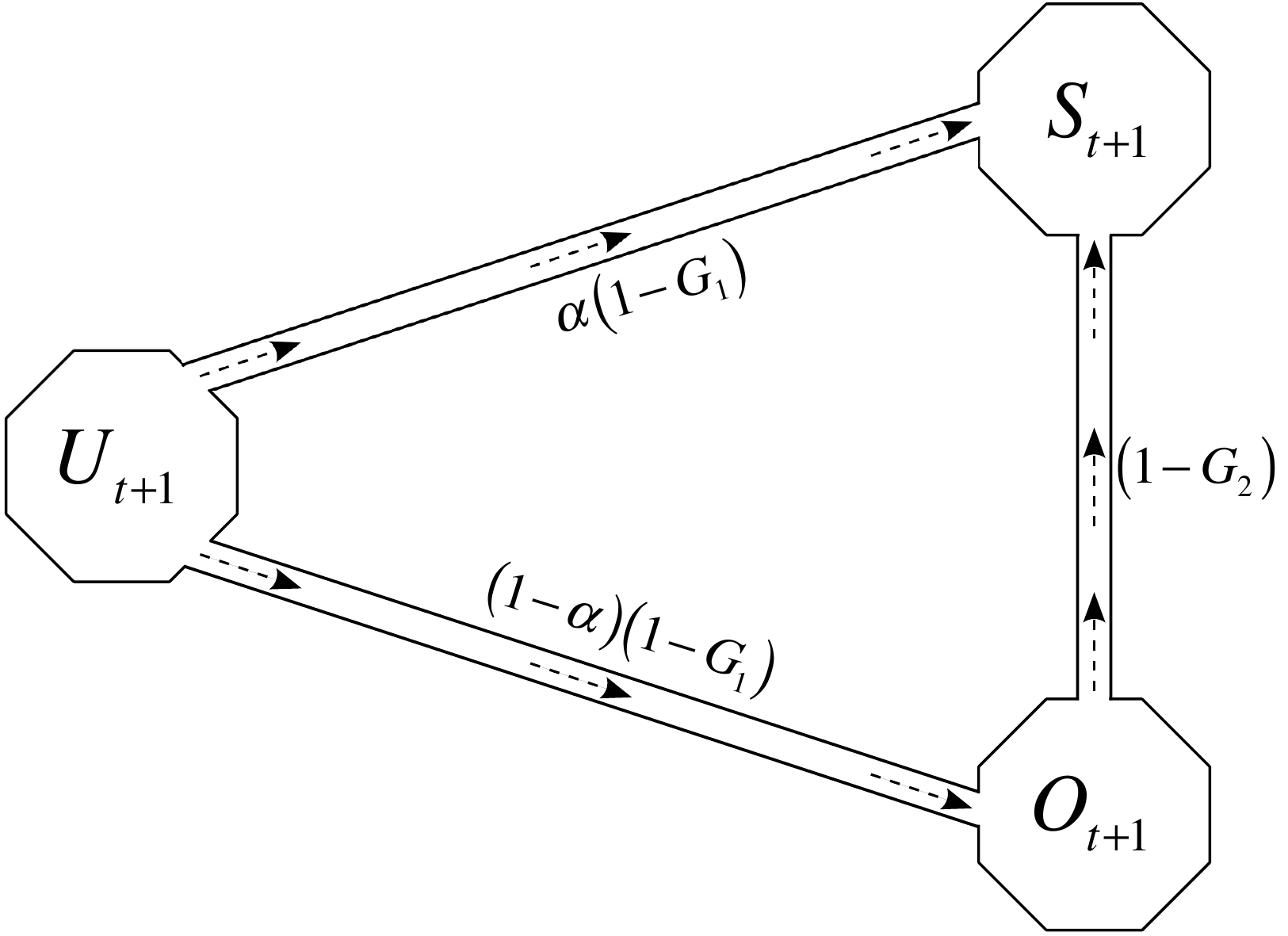} 
\caption{\label{Model:BasicComp}Compartmental Model Structure for Basic Twitter Model} 
 \label{flow_chart1}
\end{figure}

\noindent 
Before we move on to the quarantine model, we examine the equilibrium points for this basic model. There are two equilibrium points:

\begin{align*}
(U_{1}^{*},\,\,O_{1}^{*},\,\,S_{1}^{*}) &= (0,\,\,N,\,\,0) \\
(U_{2}^{*},\, O_{2}^{*},\, S_{2}^{*}) &= (U, \,\,0, \,\, S) 
\end{align*}

The first equilibrium point represents the case where there are only offenders, because the only way for an offender to leave the class is to come in contact with a stifler, of which there are none.
The second equilibrium point is analogous to the disease-free equilibrium because there is no one in the Offender class. We refer to this point as the  offender-free equilibrium (OFE). 

The Jacobian for this basic model is:
$$J = \left(
 \begin{array}{ccc}
 G_{1,t} & -\frac{k}{N} G_{1,t} U_{t} & 0 \\
 (1 - \alpha) (1 - G_{1,t}) & \frac{k}{N}(1 - \alpha) G_{1,t} U_{t} + G_{2,t} & -\frac{k}{N} G_{2,t} O_{t} \\
 \alpha (1 - G_{1,t}) & \frac{k}{N} \alpha G_{1,t} U_{t} + (1 - G_{2,t}) & \frac{k}{N} G_{2,t} O_{t} + 1 \\
 \end{array}
 \right).$$

\noindent The eigenvalues of the Jacobian corresponding to the first equilibrium point, ($0,\,\,N, \,\,0$), which we refer to the offender-only equilibrium, are:
$$\Lambda_{OOE} = \left(
 \begin {array}{c}
  1 \\
  e^{-k} \\
  k + 1 \\
\end {array}
\right).
$$
We note that $k>0$ and thus the third eigenvalue is always greater than 1.  We thus conclude that the
first equilibrium point is not stable.

Evaluating the Jacobian at the offender-free equilibrium, ($U$, $0$, $S$), we obtain 
\begin{equation}
J_{OFE} = 
\left(\begin{array}{ccc}
  \frac{k}{N} (1 - \alpha) U_{t} + G_{2,t} & 0 & 0 \\ 
  \frac{k}{N} U_{t} & 1 & 0 \\
  \frac{k}{N} U_{t} + (1 - G_{2,t}) & 0 & 1 
\end{array}\right), \label{JOFE}
\end{equation} 
with eigenvalues 
$$\Lambda_{OFE} = 
\left(  \begin {array}{c} 
 \frac{k}{N} (1 - \alpha) U_{t} + G_{2,t} \\
 1 \\
 1 \\
\end {array} 
\right).$$
Note that if $U=0$ and $S=N$ the first 
eigenvalue is $\Lambda_{(1)} = e^{-k}$ which is always greater than zero and less than 1 since $k>0$. Thus there
always exists at least one value of $(U,S)$ such that the first eigenvalue is less than 1.
However, the other two eigenvalues, $\Lambda_{(2,3)}$ are equal to 1, which does not {\it a priori}
indicate that the equilibrium point is unstable, but neither does it imply stability.

Eigenvalues equal to one can result from system constraints not being applied. We note that
one of the equations in System~\ref{eqn:basic} can be removed if we apply the constraint
$U_t+S_t+O_t=N$.  Evaluating the eigenvalues of the Jacobian of this reduced system about the
offender-free equilibrium yields
$$\Lambda_{OFE} = 
\left(  \begin {array}{c} 
 \frac{k}{N} (1 - \alpha) U_{t} + G_{2,t} \\
 1 \\
\end {array} 
\right).$$
The Jury Criterion is an algorithmic method based
upon the characteristic equation of the Jacobian that is used
to assess the stability of discrete systems\cite{allen2007}. 
%@book{allen2007introduction,
%  title={An introduction to mathematical biology},
%  author={Allen, Linda JS},
%  year={2007},
%  publisher={Pearson/Prentice Hall Upper Saddle River, NJ}
%}
A prior condition for use of the method is that the characteristic equation $D(\Lambda)$, must
be greater than zero when $\Lambda=1$.  This condition is not met for this model, thus we cannot
apply the Jury criterion.
Thus, in order to assess the stability of this equilibrium point, we would have to resort to 
higher-order methods, rather than just linearization of the system about the offender-free equilibrium
\cite{richard1996first}.  This is beyond the scope of this paper, and we thus conclude that
that the stability of the offender-free equilibrium in this model is
currently indeterminate.

To calculate the Next Generation Matrix of the model, we use the methods of Allen and van den Driessche, 2008~\cite{Allen2008}, 
and we reorder the columns and rows of the offender-free Jacobian in Equation~\eqref{JOFE} 
such that the Offender (infected) class comes first, and the Uninformed and Stifler (uninfected) classes follow.
Then, as in \cite{Allen2008}, we identify the components of this Jacobian with the following form:
$$ J_{OFE} = \left(
\begin{array}{cc}
  F + T & 0_{2x1} \\
  A & C
\end{array}
\right).
$$
Then the matrix $F$ is given by
$$ F =  
\begin{array}{c}
 \frac{k}{N} (1- \alpha) U_{t},
\end{array}
$$
the matrix $T$ is
$$T = G_{2,t},$$
and the matrix $C$ is
$$ C =  \left(
\begin{array}{cc}
 1 & 0 \\
 0 & 1
\end{array}
\right).$$

The Next Generation Matrix (NGM) is defined by $F(I-T)^{-1}$.
The spectral radius of this matrix is
\begin{eqnarray}
 {\cal Z} = \frac{ \frac{k}{N} (1 - \alpha) U_{t}}{1 - G_{2,t}}.\label{RBasic}
\end{eqnarray}

We note that the spectral radius of the matrix $C$, $\rho(C)$, is exactly equal to one; in order for the spectral radius
of the NGM to be identified with the reproduction number of the system, $\rho(C)$ must be less than one~\cite{Allen2008}.
We thus do not identify ${\cal Z}$ with the reproduction number, but note with interest that it
matches the expression in Equation~\eqref{eqn:R} that we had shown could be used
as a threshold expression for $O_{t+1}<O_{t}$.

%We can recast this equation where we assume the fraction of the population that are stiflers is $f = \frac{S}{N}$, implying $\frac{U}{N} = 1 - f$, yielding $\mathcal{T} = \frac{k(1-\alpha)(1-f)}{1-e^{-kf}}$. In Lemma 1 in Appendix A, we show that for given $k$ and $\alpha$
%there always exists a value of $f$ between 0 and 1 such that $R$ is less than one.
%By solving the transcendental equation that results when $R=1$, we can show that for any initial conditions, $R<1$. This is accomplished by Lemma \ref{lemma:revisedModel} in Appendix \ref{appendixA}, when $\lambda=1$.

%%%%%%%%%%%%%%%%%%%%%%%%%%%%%%%%%%%%%%%%%%%%%%%%%%%%%%%%%%%%%%%%%%%%%%%%%%%%%%%%%%
%%%%%%%%%%%%%%%%%%%%%%%%%%%%%%%%%%%%%%%%%%%%%%%%%%%%%%%%%%%%%%%%%%%%%%%%%%%%%%%%%%
%%%%%%%%%%%%%%%%%%%%%%%%%%%%%%%%%%%%%%%%%%%%%%%%%%%%%%%%%%%%%%%%%%%%%%%%%%%%%%%%%%
\subsubsection{Twitter Model With Quarantine Class}
To examine how isolation of offending users would hinder the spread of offensive tweets, we extend the basic model in Equation \eqref{eqn:basic} by adding a Quarantine class where the parameter $1-\lambda$ is the probability that an offender becomes quarantined. 
We also include a parameter $1-\mu$ that describes the probability that a person in the Quarantine class becomes a stifler. When $\mu$ is close to one the offenders spend longer in the Quarantine class then when $\mu$ is close to zero.
The discrete-time equations of this extended model are:
\begin{align}
U_{t+1} & = U_{t} G_{1,t} \nonumber \\
O_{t+1} & = (1 - \alpha) (1-G_{1,t}) U_{t} + \lambda G_{2,t} O_{t} \nonumber \\
Q_{t+1} & = (1- \lambda) G_{2,t} O_{t} + \mu Q_{t} \nonumber \\
S_{t+1} & = S_{t} + \alpha (1-G_{1,t}) U_{t} + (1-G_{2,t}) O_{t} + (1-\mu) Q_{t},
\label{eqn:extended} 
\end{align}
where $G_{1,t}$ and $G_{2,t}$ are as in the basic model,
\begin{align}
G_{1,t} & = e^{-\frac{k}{N} O_{t}} \nonumber \\
G_{2,t} & = e^{-\frac{k}{N} S_{t}}. \nonumber
\end{align}
The parameters of the model are shown in Table \ref{Tab:params} and the compartmental diagram is shown in Figure \ref{flow_chart2}.

\begin{figure}[h!]% figure placement: here, top, bottom, or page
 \centering
 \includegraphics[height=2.5in]{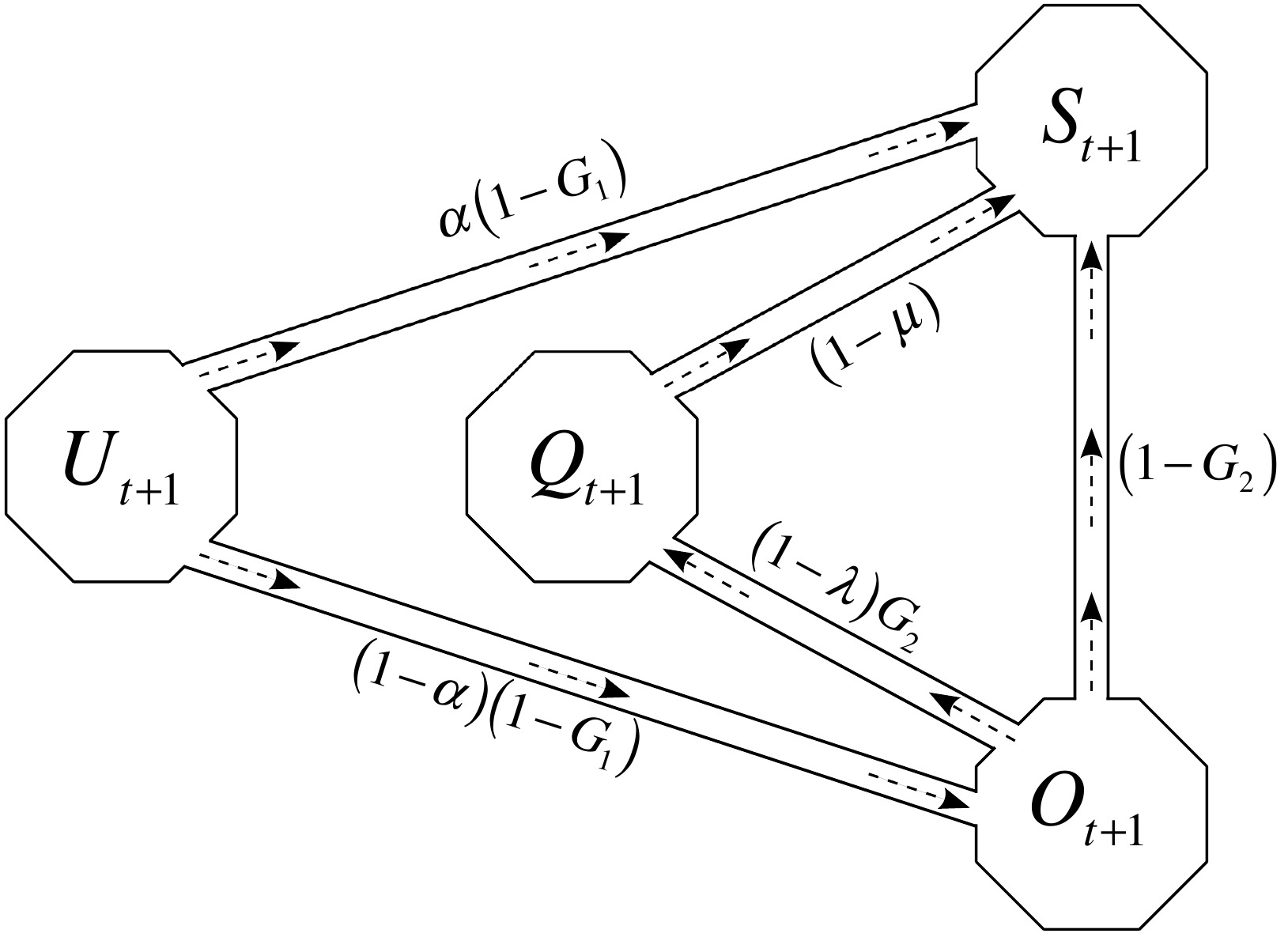} 
\caption{\label{Model:QuarComp}Compartmental Model Structure}
 \label{flow_chart2}
\end{figure}

As before in the basic model, the sum of the populations in each of the classes
adds to a constant, $N$. We also note that the population in the uninformed class monotonically asymptotically
decreases to zero, and that
the population in the Stifler class monotonically asymptocally increases to $N$.

In Theorem~\ref{thm:thresh_quar} in Appendix A, we show that if 
\begin{eqnarray}
  Z = {{k (1-\alpha) U_t}\over{N (1- \lambda G_{2,t})}} < 1
\label{eqn:Rquar}
\end{eqnarray}
then $O_{t+1} < O_t$.  
Since $U_t$ monotonically decreases to zero, and $(1-\lambda e^{-kS_t/N})$ monotonically
increases because $S_t$ monotonically increases, we see that $Z$ montonically decreases to zero.
Thus, just as in the basic model,
there will always exist a point at which $Z<1$, and $O_{t+1}<O_{t}$.  
However, we note that the smaller the 
$\lambda$ (ie; the higher the quarantine fraction $(1-\lambda)$), the
faster $Z$ will decrease to 0.  

In Theorem~\ref{thm:threshb} in Appendix A, we show additionally show that if 
\begin{eqnarray}
  Z = {{k (1-\alpha)}\over{(1- \lambda)}} < 1
\label{eqn:Rquarb}
\end{eqnarray}
then $O_{t+1}<O_t$ regardless of the values of $(U,O,Q,S)$

Further, we show in Theorem~\ref{thm:threshc} in Appendix A that if $k(1-\alpha)<1$ and
${\lambda\le 1-k(1-\alpha)}$, then ${O_{t+1}<O_t}$, regardless of the current
values of $(U,O,Q,S)$.  
We note that $k=10$ and $(1-\alpha)\sim0.001$ for
Twitter data, thus $k(1-\alpha)<1$. Thus for Twitter data there exists a 
quarantine fraction 
$\lambda$ such that the spread of offensive messages will be guaranteed to decline
at the next time step.

%In Lemma 1 in Appendix A, we show that for given values of $k$, $\alpha$ and $\lambda$ there exists an $f=S_t/N$ such that 
%$\mathcal{Z}$ is less than one (that is to say, if there are enough stiflers in the population, the rumor will begin to die out).

%Note that the threshold expression in (\ref{eqn:R}) is independent of the probability $1 - \mu$ of a quarantined user becoming a stifler. We interpret (\ref{eqn:R}) as the quarantined offenders learning their lesson and not ofending again.
%By solving the transcendental equation that results when $R=1$, we can show that for any initial conditions there exists a quarantine factor $(1-\lambda)$ such that the spread of the message can be stifled, i.e. $R<1$. This is accomplished by Lemma \ref{lemma:revisedModel} in Appendix \ref{appendixA}.

There are two equilibrium points for the model with the Quarantine class:
\begin{eqnarray}
(U_{1}^{*},\, O_{1}^{*},\, Q_{1}^{*},\, S_{1}^{*}) &=& \left(0,\,\, \frac{\lambda(\mu - 1) Q}{\lambda - 1},\,\, Q,\,\, \frac{\ln(\lambda) N}{k}
\right)\\
(U_{2}^{*},\, O_{2}^{*},\, Q_{2}^{*},\, S_{2}^{*}) &=& (U,\, 0,\, 0,\, S) 
\end{eqnarray}

\noindent The first equilibrium point is biologically irrelevant because it implies that the Stifler class will have a negative population since $0 \le \lambda < 1$. Therefore we disregard this equilibrium point. The second equilibrium point is analogous to the offender-free equilibrium of Equation \eqref{eqn:basic}. 

The Jacobian of the system is:
$$J = \left(
\begin{array}{cccc}
 G_{1,t} & -\frac{k}{N} G_{1,t} U_{t} & 0 & 0 \\
 (1- \alpha) (1-G_{1,t}) & \frac{k}{N} (1 - \alpha) G_{1,t} U_{t} + \lambda G_{2,t} & 0  & -\frac{k}{N} \lambda O_{t} G_{2,t} \\
 0 & (1- \lambda) G_{2,t} & \mu  & \frac{k}{N} O_{t} (1 - \lambda) G_{2,t} \\
 0 & \frac{k}{N} \alpha G_{1,t} U_{t} + (1-G_{2,t}) &  (1-\mu)  & 1 + \frac{k}{N} G_{2,t} O_{t} \\
\end{array}
\right). $$\\

\noindent The eigenvalues of the Jacobian evaluated at the offender-free equilibrium are 
 $\Lambda_{1,2} = 1$, $\Lambda_{3} = \mu$, and $\Lambda_{4} = (\mu + \lambda G_{2,t}) + \frac{k}{N} (1 - \alpha) U_{t}$.

As in Section~\ref{Section:Basic}, we note that the constraint $U_t+O_t+Q_t+S_t$ can be used to reduce the number of equations by one.  The
eigenvalues of the reduced system are
 $\Lambda_{1} = 1$, $\Lambda_{2} = \mu$, and $\Lambda_{3} = (\mu + \lambda G_{2,t}) + \frac{k}{N} (1 - \alpha) U_{t}$.
We note that if $k(1-\alpha)<1$ (which it is for Twitter data, which has $(1-\alpha)\sim0.001$ and $k=10$, as seen
in the previous section) then $\mu$ and $\lambda$ can always be chosen to ensure that the second and third eigenvalues are less than 1.
However the first eigenvalue is exactly equal to one.  As before, we find that the Jury criterion cannot be applied to assess the
stability of this equilibrium point, so we conclude that we find the stability to be currently indeterminate.

As in Section \ref{Section:Basic}, we begin computation of the  Next Generation Matrix for this system by reordering the equations
such that the infected classes (Offender and Quarantine) classes come first, and the uninfected (Uninformed and Stifler) classes 
follow\cite{Allen2008}.
The Jacobian of the rearranged equations, evaluated at the disease-free equilibrium is thus:
$$J_{OFE} = \left( \begin{array}{cccc}
+{{k}\over{N}} U (1-\alpha)+\lambda G_{2,t} & 0 & 0 & 0 \\
(1-\lambda)G_{2,t} & \mu & 0 & 0 \\
-{{k}\over{N}} U & 0 & 1 & 0\\
+{{k}\over{N}} \alpha U + (1-G_{2,t}) & (1-\mu) & 0 & 1 \\
\end{array} \right)$$

\noindent As before, we identify the components of this matrix with \cite{Allen2008}
$$ J_{OFE} = \left(
\begin{array}{cc}
  F + T & 0_{2x1} \\
  A & C
\end{array}
\right).
$$
Then the matrix $F$ is given by
$$ F = 
\left(
\begin{array}{cc}
+{{k}\over{N}} U (1-\alpha) & 0 \\
(1 - \lambda) G_{2,t} & 0 
\end{array}
 \right),
 $$
matrix $T$ is
$$ T =
\left(
\begin{array}{cc}
\lambda G_{2,t} & 0 \\
0 & \mu
\end{array}
 \right),
 $$
and the matrix $C$ is 
$$ C =
\left(
 \begin{array}{cc}
1 & 0 \\
0 & 1
\end{array} \right).
$$

\noindent The Next Generation Matrix is defined by $F(I-T)^{-1}$.
The spectral radius of this matrix is
\begin{eqnarray}
\mathcal{Z} = \frac{k (1 - \alpha) U_{t}}{N (1 - \lambda G_{2,t})}.
\label{eqn:Rquar}
\end{eqnarray}

We again note that the spectral radius of the matrix $C$, $\rho(C)$, is exactly equal to one
and thus we do not identify $Z$ with the reproduction number, but again note with interest that it
matches the expression in Equation~\eqref{eqn:Rquar} that we have shown could be used
as a threshold condition for $O_{t+1}<O_{t}$ in the Quarantine model.

%Just as we did in the previous section, we can recast this equation where we assume the fraction of the population that are stiflers is $f = \frac{S}{N} \Rightarrow \frac{U}{N} = 1 - f \Rightarrow \frac{k(1-\alpha)(1-f)}{1 - \lambda e^{-kf}}$. 
%
%We note that if $k(1-\alpha)<1$ the numerator $k(1-\alpha)(1-f)$ in (\ref{eqn:R}) is always less than 1.  In the spread of offensive messages on Twitter, we found that $k = 10$ and $(1-\alpha) = 0.001$,
%thus in the particular case of interest to this
%analysis, $k(1-\alpha)=10\ast0.001=0.01<1$. For $k(1-\alpha) < 1$, then $0 \leq \lambda \leq 1$ or $0 < \lambda < 1$ such that $k(1-\alpha)(1-f)<(1-\lambda e^{-kf})$ (i.e. $\mathcal{T} < 1$), regardless of the value of $f$.  
%This implies that with the spread of offensive messages on Twitter, there
%can always be a sufficiently large fraction of people quarantined that will prevent the
%number of offenders from becoming larger at the next time step.

\begin{table}[ht]
\begin{tabular}{| c | c | c | } \hline
\textbf{Symbol} & \textbf{Definition} & \textbf{Value}\\ \hline

$\bm{N} $ & Number of users in the population & $10^{7}$\\ \hline
$\bm{U_{t}}$ & Number of uninformed users at time $t$ & Varies\\ \hline
$\bm{O_{t}}$ & Number of offenders at time $t$ & Varies\\ \hline
$\bm{Q_{t}}$ & Number of quarantined users at time $t$ & $0$ at $t=0$\\ \hline
$\bm{S_{t}}$ & Number of stiflers at time $t$ & Varies\\ \hline
$\bm{(1 - G_{1,t})}$ & Probability an uninformed user sees the offensive tweet & Based on $O_{t} \,\& \,Q_{t}$\\ \hline
$\bm{(1 - G_{2,t})}$ & Probability an offender becomes a stifler & Based on $S_{t}$\\ \hline
\multirow{2}{*}{$\bm{(1 - \alpha)}$} & Probability an uninformed user becomes  &
\multirow{2}{*}{$0.001$} \\
 & an offender (after seeing the offensive tweet) & \\ \hline
\multirow{2}{*}{$\bm{\alpha}$} & Probability an uninformed user becomes  & \multirow{2}{*}{$0.999$} \\
  & a stifler (after seeing the offensive tweet) & \\ \hline
$\bm{(1 - \lambda)}$ & Probability an offender becomes quarantined & Varies\\ \hline
$\bm{(1 - \mu)}$ & Probability a quarantined user becomes a stifler & Varies\\ \hline
$\bm{k}$ & The average degree of the network  & 10 \\ \hline
\end{tabular}
\centering
\caption{\label{Tab:params}
Parameter values for the discrete-time models.}
\end{table}
%%%%%%%%%%%%%%%%%%%%%%%%%%%%%%%%%%%%%%%%%%%%%%%%%%%%%%%%%%%%%%%%%%%%%%%%%%%%%%%%%%%%%%%%%%

%%%%%%%%%%%%%%%%%%%%%%%%%%%%%%%%%%%%%%%%%%%%%%%%%%%%%%%%%%%%%%%%%%%%%%%%%%%%%%%%%%%%%%%%%%
\section{Results}
We determined the parameters of Equation \eqref{eqn:extended} using current Twitter data:
\begin{itemize}
\item $\alpha$, the probability that the user becomes a stifler
\item $\frac{S_{0}}{N}$ the fraction of the population that is initally in the Stifler class.
\item $k$, the degree of the network.
\end{itemize}
We assumed based on our data and information from the literature that $\alpha = 0.999$ and $k = 10$ as described in Section \ref{sec:Data}. We examined values of $\frac{S_{0}}{N}$ between $0.01$ to 1 in steps of $0.01$ and determined the model prediction for the number of retweets, and compared this to the number of retweets we observed in the data (namely approximately one retweet per tweet). This can be seen in Table \ref{Tab:params}. We found the value of $\frac{S_{0}}{N} = 0.07$ best fit the data. We let $N = 10,000,000$ but we found our results were not sensitive to the value of $N$, for large $N$. Figures \ref{fg:Basic_Model_Output} and \ref{fg:Quarantined_Model_Output} are numerical simulations of our basic and quarantined models respectively. These Figures visually show the dynamics of the systems.

\begin{figure}[h!]
\centering
\includegraphics[scale=0.5]{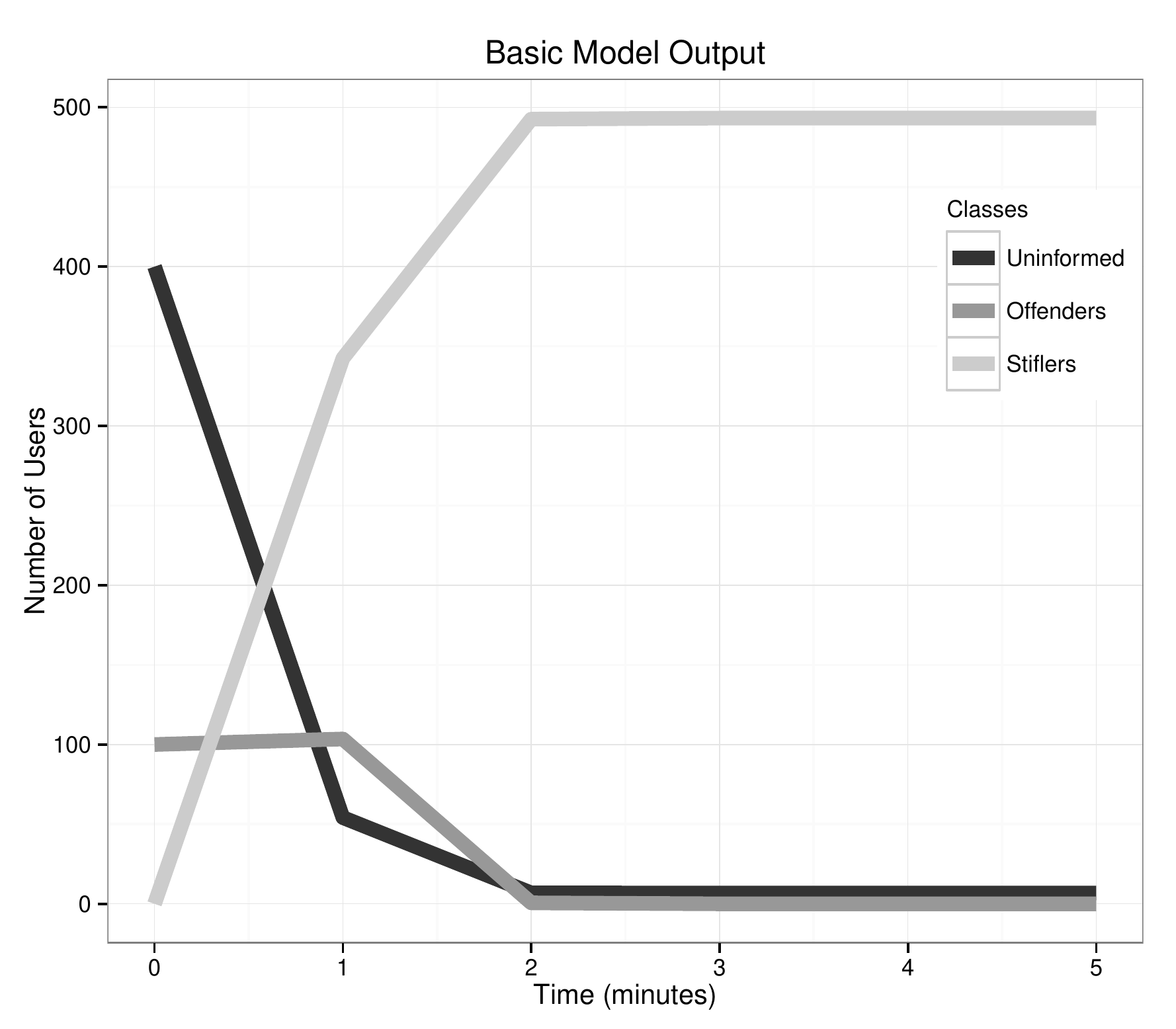}
\caption{This is the basic output of the model without quarantining. The parameters are $\alpha = 0.999$, $k = 10$. the Uninformed and Offender classes tend to zero while the Stifler class tends to $N$ which is the total population.
\label{fg:Basic_Model_Output}}
\end{figure}

\begin{figure}[h!]
\centering
\includegraphics[scale=0.5]{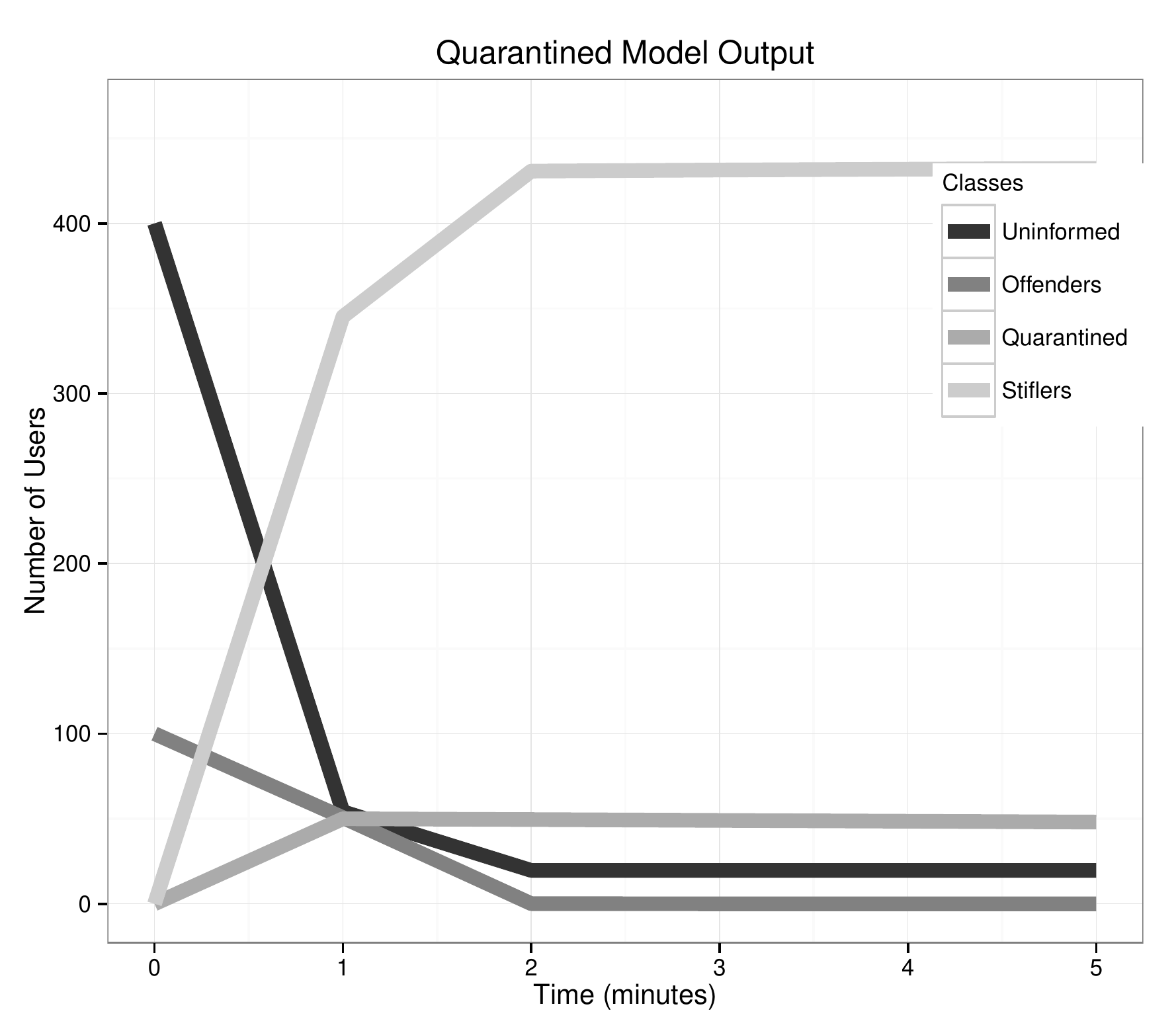}
\caption{This is the model with Quarantine class. The parameters are $\alpha = 0.999$, $k = 10$, $\mu = 0.99$, $\lambda = 0.5$. The Uninformed, Offender, and Quarantined classes tend to zero while the Stifler class tends to $N$ which is the total population.
\label{fg:Quarantined_Model_Output}}
\end{figure}

 We used the parameter values in Table \ref{Tab:params} to find solutions to Equation \eqref{eqn:extended}. In our computations, we assume the probability a quarantined user becomes a stifler, $(1 - \mu)$ is zero. To do this we examined values of $(1-\lambda)$, the quarantine fraction, between zero and one and obtained the model prediction for the number of retweets. The results are seen in Figure \ref{fg:LambdaVsThreshold}.

\begin{figure}[h!]
\centering
\includegraphics[scale=0.9]{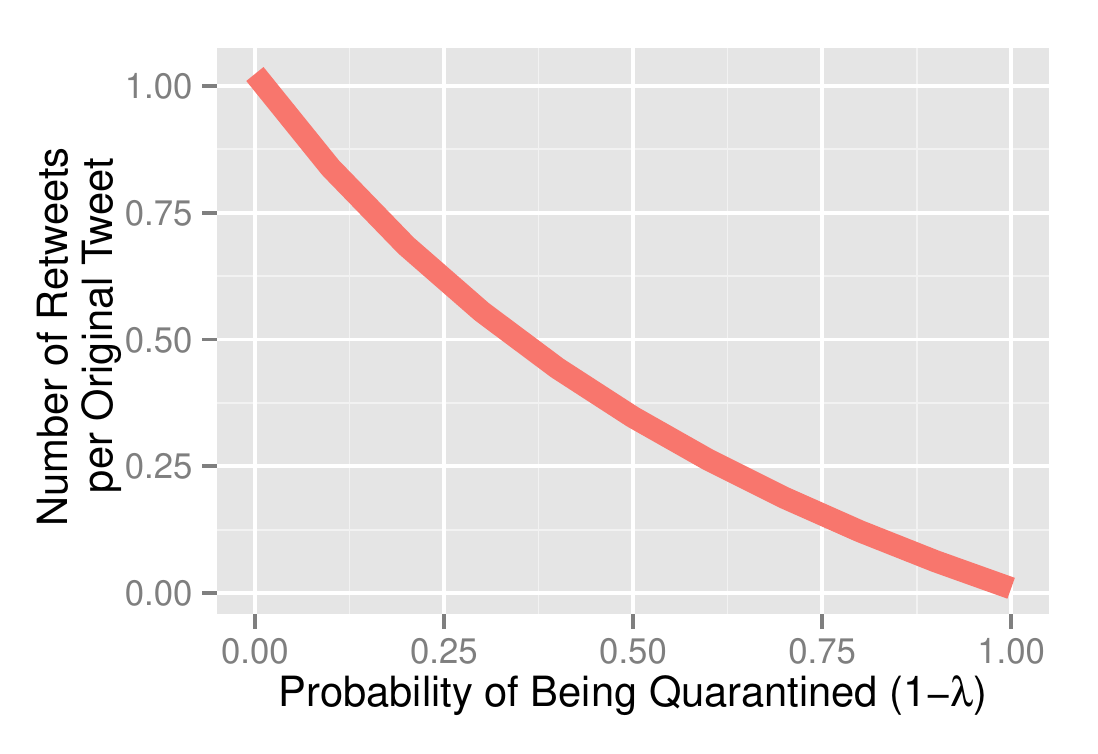}
\caption{ Prediction of the quarantine model for the number of retweets per original tweet as a function of the quarantine fraction, $(1-\lambda)$. For $(1-\alpha) = 0.001$, $k=10$, $S_0/N=0.07$. We assume that quarantined users are fully isolated from the rest of the community. As the probability of quarantine increases, the number of retweets decreases.
\label{fg:LambdaVsThreshold}}
\end{figure}

By looking at the number of offenders over time, while varying $\lambda$ we can see how the quarantine fraction, $(1-\lambda)$, affects the life-time of the tweet. As $(1-\lambda)$ increases, the model predicts a decrease in the average life-time of the tweet, see in Figure \ref{fg:OffendersVsTimeLambda}, i.e. the tweet on average dies out more quickly.

\begin{figure}[h!]
\centering
\includegraphics[scale=0.6]{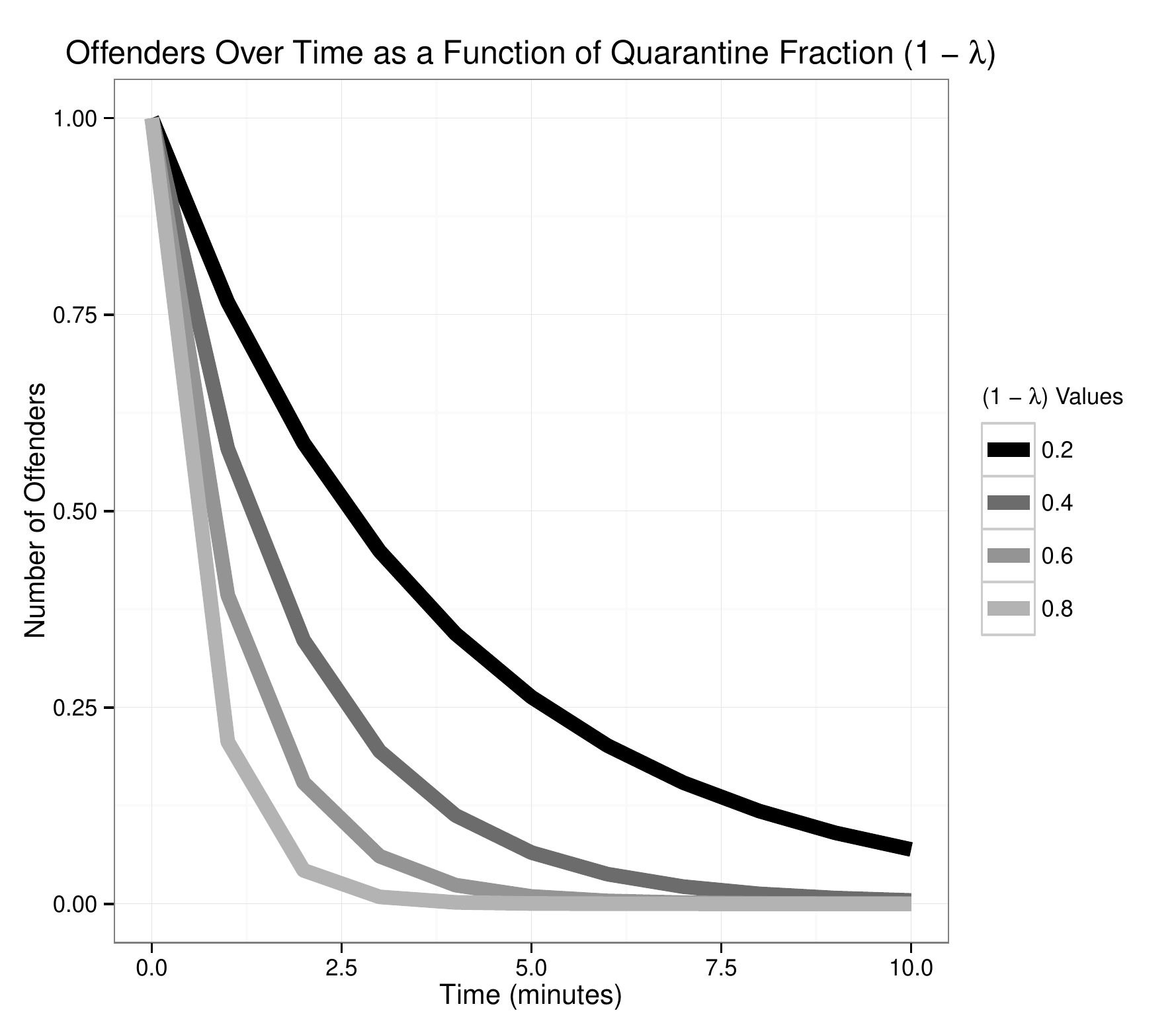}
\caption{Comparison of the number of offenders over time as $1 - \lambda$ varies. the parameters $\alpha = 0.999$, $\mu = 0.99$, and $k = 10$. For a low probability of being quarantined, there is a high number of offenders over time. As the probability of the quarantined increases, the offenders deceases}
\label{fg:OffendersVsTimeLambda}
\end{figure}

%To examine how variation of $(1-\sigma)$, the probability that a quarantined user comes in contact with an uninformed user, affects the number of offenders we repeat our model simulation for various values of $\sigma$ and the results are shown in Figure \ref{fg:OffendersVsTimeSigma}.
%
%\begin{figure}[ht]
%\centering
%\includegraphics[scale=0.7]{Offenders_vs_Sigma.pdf}
%\caption{This graph shows the number of offenders over time as $1 - \sigma$ varies. We let $(1-\lambda) = 0.5, \alpha = 0.99, \mu = 0.99$ and $k = 10$ For a low probability of contacting a quarantined user, there is a low number of offenders.}
%\label{fg:OffendersVsTimeSigma}
%\end{figure}
\clearpage

\section{Discussion}
In this paper we presented a discrete-time non-linear compartmental model to simulate the dynamics of the spread of offensive messages on Twitter. We developed the model to gain insight into reducing the number of offensive tweets seen by Twitter users. We then extended the model to include a Quarantine class whereby offending users were isolated from contact with the rest of the Twitter community. 
We found that the implementation of quarantine was an effective way to quickly reduce the spread
of offensive messages.  Further, we proved that if $k(1-\alpha)<1$, where $k$ is the degree of the network,
and $(1-\alpha)$ is the probability of retweeting, then if $\lambda<1-k(1-\alpha)$ the number of offenders
at the next time step will always be less than the number of offenders the time step before.  
We showed that this was independent of the the number of uninformed, offenders, quarantined, and
stiflers at that time step. 
From data we obtained from Twitter, 
and the published literature, we determined that $k=10$ and $(1-\alpha)\sim 0.001$,
thus for Twitter $k(1-\alpha)<1$ and there thus exists a $\lambda$ such that offensive tweets will be 
effectively suppressed.   

%We performed mathematical analysis of the stability of the equilibrium points and the threshold values of both model. We found that the implementation of quarantine was an effective way to reduce the spread of offensive messages. 
We performed numerical simulations to examine the effect of quarantine.
As shown in Figure 3, the quarantine of a certain fraction of offending Twitter users led to a disproportionately larger reduction in the number of new offenders. 

Previous studies in this area have used a wide variety of statistical, computational, and mathematical methods to understand the spread of information in social networks (see \cite{Budak2011,Wang2012, Han2012,Nekovee2007,Tripathy2010}, for instance). The use of a Quarantine class to hinder the spread of messages has not been previously considered.

There are some limitations of the model we have developed. For instance, we assume that once offenders leave the Quarantine class, they have learned their lesson and do not offend again. This is not realistic in all cases. Also, while we find that the quarantine of offenders is effective, we have not determined how such a quarantine can be implemented. One way to potentially implement the quarantine is to have Twitter use an algorithm that flags tweets as being potentially offensive based on keywords \cite{Xu2012}. Another way to implement a quarantine is to have users give negative rating to tweets they find offensive, for instance, a rating system that goes from zero to minus ten by integer values and lets each user choose the value that fits the type of tweets they want to see.

%Further Research
With this basic and quarantine model, there are other routes that could be taken to further this research. For instance, performing an analysis of a Markov chain model. We can also examine how quarantined users or stiflers stop the spread of a tweet through a network graph. Performing a network analysis on such a model would provide insight into which users should be quarantined first, or which groups to isolate from the overall population, to reduce the spread of these tweets. 

In conclusion, our model indicates that quarantine of users that spread offensive message would likely be effective. Our use of a discrete-time non-linear model highlights the usefulness of mathematical models in understanding the underlying dynamics of social networks, and opens up many interesting possibilities for future research.

\newpage
\section*{Acknowledgments} 
We would like to thank Dr.~Carlos Castillo-Chavez, Executive Director of the Mathematical and Theoretical Biology Institute (MTBI), for giving us the opportunity to participate in this research program.  We would also like to thank Co-Executive Summer Directors Dr.~Erika T.~Camacho and Dr.~Stephen Wirkus for their efforts in planning and executing the day to day activities of MTBI. We also want to give special thanks to Daniel Burkow for helping us redefine our topic, Kamal Barley for his graphic enhancements, Dr.~Luis Melara, Dr.~Leon Arriola, Xiaoguang Zhang, and Oyita Udiani for their invaluable help and feedback in writing this paper. This research was conducted in MTBI at the Mathematical, Computational and Modeling Sciences Center (MCMSC) at Arizona State University (ASU). This project has been partially supported by grants from the National Science Foundation (NSF - Grant DMPS-1263374), the National Security Agency (NSA - Grant H98230-13-1-0261), the Office of the President of ASU, and the Office of the Provost of ASU.

%%------------------------------------------------------------------------%%
%%-----------------------Appendix------------------------------------%%
%%------------------------------------------------------------------------%%
\newpage
\section{Appendix A} 
\label{appendixA}

%%%%%%%%%%%%%%%%%%%%%%%%%%%%%%%%%%%%%%%%%%%%%%%%%%%%%%%%%%%%%%%%%%%%%%%%%%%%%%%%%%
%%%%%%%%%%%%%%%%%%%%%%%%%%%%%%%%%%%%%%%%%%%%%%%%%%%%%%%%%%%%%%%%%%%%%%%%%%%%%%%%%%
% proof that shows in basic models the class populations are never negative
%%%%%%%%%%%%%%%%%%%%%%%%%%%%%%%%%%%%%%%%%%%%%%%%%%%%%%%%%%%%%%%%%%%%%%%%%%%%%%%%%%

%%%%%%%%%%%%%%%%%%%%%%%%%%%%%%%%%%%%%%%%%%%%%%%%%%%%%%%%%%%%%%%%%%%%%%%%%%%%%%%%%%
%%%%%%%%%%%%%%%%%%%%%%%%%%%%%%%%%%%%%%%%%%%%%%%%%%%%%%%%%%%%%%%%%%%%%%%%%%%%%%%%%%
% proof that shows in basic models the class populations are never negative
% if they start out non-negative
%%%%%%%%%%%%%%%%%%%%%%%%%%%%%%%%%%%%%%%%%%%%%%%%%%%%%%%%%%%%%%%%%%%%%%%%%%%%%%%%%%
\begin{lem}
\label{lem:gt0}
If $\mmu_{0}$, $\mmo_{0}$, $\mms_{0} \ge 0$  in the basic model in System~\ref{eqn:basic},
then $\ut$, $\ot$, $\st \ge 0$
\end{lem}
\begin{proof}
Let $\mmu_{0}$, $\mmo_{0}$, $\mms_{0} \ge 0$
\begin{align}
\mathcal{U}_{1} = \mmu_{0} e^{\frac{-k}{N}\mmo_{0}} &\\
\mmu_{2} = \mmu_{1} e^{\frac{-k}{N}\mmo_{1}} & = \mmu_{0} e^{\frac{-k}{N}(\mmo_{0}+\mmo_{1})}\\
\vdots \\
\mmu_{n+1} = \mmu_{n} e^{\frac{-k}{N}\mmo_{n}} & = \mmu_{0} e^{\frac{-k}{N}(\mmo_{0}+\mmo_{1}+\mmo_{2}+\dots+\mmo_{n})}
%0 < \uto \le \ut \le \mathcal{U}_{0}
\end{align}
Since $e^{\frac{-k}{N}(\sum_{i}^{n} \mmo_{i})} > 0$, we can see that $\ut \ge 0 $ for all $t$.  Therefore, as $\oto \ge \ot \gs$, we can use a similar argument to show that for $\mmo_{0} \ge 0$
\begin{equation}
\oto \ge \ot \gs = \mmo_{0} e^{\frac{-k}{N}(\mms_{0}+\mms_{1}+\mms_{2}+\dots+\mms_{t})}
\end{equation}
Lastly to show that $\st \ge 0$ for all $t$ we notice that $\sto \ge \st$ since $\ut$ and $\ot$ are always non-negative. Therefore since $\mms_0 \ge 0$, and $\sto \ge \st$ then $\st \ge 0$
\end{proof}

%%%%%%%%%%%%%%%%%%%%%%%%%%%%%%%%%%%%%%%%%%%%%%%%%%%%%%%%%%%%%%%%%%%%%%%%%%%%%%%%%%
%%%%%%%%%%%%%%%%%%%%%%%%%%%%%%%%%%%%%%%%%%%%%%%%%%%%%%%%%%%%%%%%%%%%%%%%%%%%%%%%%%
\begin{lem}
\label{lem:bound}
The basic model given by System~\ref{eqn:basic}
is bounded such that the sum of the compartments remains constant.
\end{lem}
\begin{proof}{}
Let $N_t = \ut + \ot + \st$;
\begin{align*}
N_{t+1} & = \uto + \oto + \sto \\
& = \ut \go + (1-\alpha)(1-\go)\ut + \ot \gs + \alpha (1 - \go)\ut + (1- \gs) \ot + \st \\
& = \ut \go + (1-\alpha)(1-\go)\ut + \alpha (1 - \go)\ut + \ot \gs + (1- \gs) \ot + \st \\
& = \ut + \ot + \st \\
& = N_t
\end{align*}
\end{proof}

% proof that in the basic model U monotonically decreases
%%%%%%%%%%%%%%%%%%%%%%%%%%%%%%%%%%%%%%%%%%%%%%%%%%%%%%%%%%%%%%%%%%%%%%%%%%%%%%%%%%
\begin{lem}
\label{lem:decU}
If $\mmu_0$, $\mmo_0 > 0$ in the basic model in System~\ref{eqn:basic}, then $0 < \uto < \ut$ for all $t$.
\end{lem}
\begin{proof}
Let $\mmu_0$, $\mmo_0 > 0$. From (7),(8) we know
\begin{align*}
& \uto = \mmu_{0} e^{\frac{-k}{N}(\mmo_{0}+\mmo_{1}+\mmo_{2}+\dots+\mmo_{t}) }\\
& \oto \ge \mmo_{0} e^{\frac{-k}{N}(\mms_{0}+\mms_{1}+\mms_{2}+\dots+\mms_{t})}
\end{align*}
Notice that if $\sum_{i=0}^{t} \mmo_{i} >0$ then $e^{\frac{-k}{N}(\sum_{i=0}^{t} \mmo_{i})} \in (0,1)$.  Since $\mmu_{0}>0$, we can see that $\ut > 0$ for all $t$ and consequently
\begin{align*}
\oto & = (1-\alpha)(1-\go)\ut + \ot \gs \\
\implies \oto & > \ot \gs > 0
\end{align*}
Therefore since $\{ \mmo_{i} \}_{i}^{t}>0$, we can say $\sum_{i=0}^{t} \mmo_{i} >0$ and $\uto < \ut$.
\end{proof}

%%%%%%%%%%%%%%%%%%%%%%%%%%%%%%%%%%%%%%%%%%%%%%%%%%%%%%%%%%%%%%%%%%%%%%%%%%%%%%%%%%
%%%%%%%%%%%%%%%%%%%%%%%%%%%%%%%%%%%%%%%%%%%%%%%%%%%%%%%%%%%%%%%%%%%%%%%%%%%%%%%%%%
% the proof for the basic model threshold value
%%%%%%%%%%%%%%%%%%%%%%%%%%%%%%%%%%%%%%%%%%%%%%%%%%%%%%%%%%%%%%%%%%%%%%%%%%%%%%%%%%
\begin{theorem}
\label{thm:thresh}
When ${Z = {{k (1-\alpha) U_t}\over{N (1-G_{2,t})}} < 1}$
then ${O_{t+1}<O_t}$ in the basic model in System~\ref{eqn:basic}. 
\end{theorem}
\begin{proof}{}
For the basic model in System~\ref{eqn:basic} we have
\begin{eqnarray}
   O_{t+1} = (1-\alpha)(1-G_{1,t})U_t + G_{2,t} O_t
\label{eqn:O}
\end{eqnarray}

\noindent Note that for all ${(kO_t/N)>0}$ the following is true
\begin{eqnarray}
  (1-e^{-kO_t/N})< kO_t/N
\end{eqnarray}
Since ${G_{1,t}=e^{-kO_t/N}}$, this implies that
\begin{eqnarray}
  (1-G_{1,t})N/k < O_t
\label{eqn:less}
\end{eqnarray}

\noindent Now, we wish to show that when
\begin{eqnarray}
  Z = {{k (1-\alpha) U_t}\over{N (1-G_{2,t})}} < 1
\label{eqn:Z}
\end{eqnarray}
then
\begin{eqnarray}
  O_{t+1} < O_t. \nonumber
\end{eqnarray}
Note that Equation \eqref{eqn:Z} implies that
\begin{eqnarray}
  (1-\alpha) U_t < {{N}\over{k}} (1-G_{2,t}). \nonumber
\end{eqnarray}
Substituting this into Equation~\eqref{eqn:O} yields
\begin{eqnarray}
O_{t+1} < {{N}\over{k}} (1-G_{1,t}) (1- G_{2,t}) + G_{2,t} O_t
\end{eqnarray}
From Equation~\eqref{eqn:less} we know that ${{{N}\over{k}} (1-G_{1,t})<O_t}$, thus
\begin{eqnarray}
O_{t+1} & < & O_t (1- G_{2,t}) +  G_{2,t} O_t
\end{eqnarray}
The LHS is equal to $O_t$.  Thus
\begin{eqnarray}
O_{t+1}<O_t.
\end{eqnarray}
\end{proof}

%%%%%%%%%%%%%%%%%%%%%%%%%%%%%%%%%%%%%%%%%%%%%%%%%%%%%%%%%%%%%%%%%%%%%%%%%%%%%%%%%%
%%%%%%%%%%%%%%%%%%%%%%%%%%%%%%%%%%%%%%%%%%%%%%%%%%%%%%%%%%%%%%%%%%%%%%%%%%%%%%%%%%
% the proof for the quarantine model threshold value
%%%%%%%%%%%%%%%%%%%%%%%%%%%%%%%%%%%%%%%%%%%%%%%%%%%%%%%%%%%%%%%%%%%%%%%%%%%%%%%%%%
\begin{theorem}
\label{thm:thresh_quar}
 When ${Z = {{k (1-\alpha) U_t}\over{N (1-\lambda G_{2,t})}} < 1}$
then ${O_{t+1}<O_t}$ in the Quarantine model.
\end{theorem}

\begin{proof}
For the Quarantine model we have
\begin{eqnarray}
   O_{t+1} = (1-\alpha)(1-G_{1,t})U_t + \lambda G_{2,t} O_t
\label{eqn:O}
\end{eqnarray}

\noindent Note that for all ${(kO_t/N)>0}$ the following is true
\begin{eqnarray}
  (1-e^{-kO_t/N})< kO_t/N
\end{eqnarray}
Since ${G_{1,t}=e^{-kO_t/N}}$, this implies that
\begin{eqnarray}
  (1-G_{1,t})N/k < O_t
\label{eqn:less}
\end{eqnarray}

\noindent Now, we wish to show that when
\begin{eqnarray}
  Z = {{k (1-\alpha) U_t}\over{N (1-\lambda G_{2,t})}} < 1
\label{eqn:Z}
\end{eqnarray}
then
\begin{eqnarray}
  O_{t+1} < O_t. \nonumber
\end{eqnarray}
Note that Equation \eqref{eqn:Z} implies that
\begin{eqnarray}
  (1-\alpha) U_t < {{N}\over{k}} (1-\lambda G_{2,t}). \nonumber
\end{eqnarray}
Substituting this into Equation~\eqref{eqn:O} yields
\begin{eqnarray}
O_{t+1} < {{N}\over{k}} (1-G_{1,t}) (1-\lambda G_{2,t}) + \lambda G_{2,t} O_t
\end{eqnarray}
From Equation~\eqref{eqn:less} we know that ${{{N}\over{k}} (1-G_{1,t})<O_t}$, thus
\begin{eqnarray}
O_{t+1} & < & O_t (1-\lambda G_{2,t}) + \lambda G_{2,t} O_t
\end{eqnarray}
The LHS is equal to $O_t$.  Thus
\begin{eqnarray}
O_{t+1}<O_t.
\end{eqnarray}

\end{proof}

%%%%%%%%%%%%%%%%%%%%%%%%%%%%%%%%%%%%%%%%%%%%%%%%%%%%%%%%%%%%%%%%%%%%%%%%%%%%%%%%%%
%%%%%%%%%%%%%%%%%%%%%%%%%%%%%%%%%%%%%%%%%%%%%%%%%%%%%%%%%%%%%%%%%%%%%%%%%%%%%%%%%%
% the proof for the quarantine model threshold value with no U and S
%%%%%%%%%%%%%%%%%%%%%%%%%%%%%%%%%%%%%%%%%%%%%%%%%%%%%%%%%%%%%%%%%%%%%%%%%%%%%%%%%%
\begin{theorem}
\label{thm:threshb}
When ${Z = {{k (1-\alpha)}\over{(1-\lambda)}} < 1}$ then ${O_{t+1}<O_t}$ in the Quarantine model.
\end{theorem}

\begin{proof}
In Theorem \ref{thm:thresh_quar} we showed that when
${Z = {{k (1-\alpha) U_t}\over{N (1-\lambda G_{2,t})}} < 1}$
then $O_{t+1}<O_t$ in the Quarantine model.
Thus it suffices here to show that
\begin{eqnarray}
{{k (1-\alpha)}\over{(1-\lambda )}}\ge{{k (1-\alpha) U_t/N}\over{(1-\lambda e^{-kS_t/N})}}
\label{eqn:obj}
\end{eqnarray}
for all $U_t$ and $S_t$.

Note that the numerator of the LHS is always less than or equal to  $k(1-\alpha)$ because ${0\le U_t/N\le1}$,
and the denominator
of the LHS is always greater than or equal to $(1-\lambda)$, because ${0\le S_t/N\le1}$ and $k>0$ thus
${e^{-kS_t/N}<1}$.  Thus we have shown that the LHS is at most the RHS. And thus
when ${Z = {{k (1-\alpha)}\over{(1-\lambda)}} < 1}$,
then ${O_{t+1}<O_t}$ in the Quarantine model.

\end{proof}

%%%%%%%%%%%%%%%%%%%%%%%%%%%%%%%%%%%%%%%%%%%%%%%%%%%%%%%%%%%%%%%%%%%%%%%%%%%%%%%%%%
%%%%%%%%%%%%%%%%%%%%%%%%%%%%%%%%%%%%%%%%%%%%%%%%%%%%%%%%%%%%%%%%%%%%%%%%%%%%%%%%%%
% the proof for the quarantine model threshold value 
% will always be less than 1 if k(1-alpha)<1 and lambda<some value 
%%%%%%%%%%%%%%%%%%%%%%%%%%%%%%%%%%%%%%%%%%%%%%%%%%%%%%%%%%%%%%%%%%%%%%%%%%%%%%%%%%
\begin{theorem}
\label{thm:threshc}
When ${k(1-\alpha)<1}$ and
${\lambda\le 1-k(1-\alpha)}$,
then ${Z = {{k (1-\alpha)}\over{(1-\lambda}} < 1}$ and ${O_{t+1}<O_t}$.
\end{theorem}

\begin{proof}
We have
$\lambda\le 1-k(1-\alpha)$, and we also have that
$k(1-\alpha)<1$, which ensures that $\lambda>0$.
We thus have that
$Z = {{k (1-\alpha)}\over{(1-\lambda }} < {{k(1-\alpha)}\over{(1-[1-k(1-\alpha)])}}$.
The RHS is equal to 1, thus we have shown that
when $k(1-\alpha)<1$ and
$\lambda\le 1-k(1-\alpha)$,
then $Z = {{k (1-\alpha)}\over{(1-\lambda}} < 1$.
\end{proof}

%%------------------------------------------------------------------------%%
%%-----------------------BIBLIOGRAPHY---------------------------%%
%%------------------------------------------------------------------------%%


\newpage
\begin{thebibliography}{10}

\bibitem{Qiu2012}
Lin Qiu, Han Lin, Jonathan Ramsay, and Fang Yang.
\newblock You are what you tweet: Personality expression and perception on
  twitter.
\newblock {\em Journal of Research in Personality}, 00 2012.

\bibitem{Gruzd2011}
A~Gruzd, B~Wellman, and Y~Takhteyev.
\newblock Imagining twitter as an imagined community.
\newblock {\em American Behavioral Scientist}, 55(10):1294--1318, 00 2011.

\bibitem{Bae2012}
Younggue Bae and Hongchul Lee.
\newblock Sentiment analysis of twitter audiences: Measuring the positive or
  negative influence of popular twitterers.
\newblock {\em Journal of the American Society for Information Science and
  Technology}, 63(12):2521--2535, 00 2012.

\bibitem{Dodds2011}
Peter~Sheridan Dodds, Kameron~Decker Harris, Isabel~M Kloumann, Catherine~A
  Bliss, and Christopher~M Danforth.
\newblock Temporal patterns of happiness and information in a global social
  network: hedonometrics and twitter.
\newblock {\em PLoS ONE}, 6(12):e26752, 00 2011.

\bibitem{Hinduja2008}
Sameer Hinduja and Justin~W. Patchin.
\newblock Cyberbullying: An exploratory analysis of factors related to
  offending and victimization.
\newblock {\em Deviant Behavior}, 29(2):129--156, 2008.

\bibitem{Kwak2010}
Haewoon Kwak, Changhyun Lee, Hosung Park, and Sue Moon.
\newblock What is twitter, a social network or a news media?
\newblock In {\em Proceedings of the 19th international conference on World
  wide web}, WWW '10, pages 591--600, New York, NY, USA, 2010. ACM.

\bibitem{Zappavigna2011}
Michele Zappavigna.
\newblock Ambient affiliation: A linguistic perspective on twitter.
\newblock {\em New Media \& Society}, 13(5):788--806, 00 2011.

\bibitem{Gorzig2013}
Anke G{\"o}rzig and Lara~A Frumkin.
\newblock Cyberbullying experiences on-the-go: When social media can become
  distressing.
\newblock {\em Cyberpsychology: Journal of Psychosocial Research on
  Cyberspace}, 7(1), 00 2013.

\bibitem{Li2006}
Qing Li.
\newblock Cyberbullying in schools: A research of gender differences.
\newblock {\em School Psychology International}, 27(2):157--170, 2006.

\bibitem{Xu2012}
Jun-Ming Xu, Kwang-Sung Jun, Xiaojin Zhu, and Amy Bellmore.
\newblock Learning from bullying traces in social media.
\newblock In {\em Proceedings of the 2012 Conference of the North American
  Chapter of the Association for Computational Linguistics: Human Language
  Technologies}, NAACL HLT '12, pages 656--666, Stroudsburg, PA, USA, 2012.
  Association for Computational Linguistics.

\bibitem{Zhao2013a}
Laijun Zhao, Hongxin Cui, Xiaoyan Qiu, Xiaoli Wang, and Jiajia Wang.
\newblock Sir rumor spreading model in the new media age.
\newblock {\em Physica A: Statistical Mechanics and its Applications},
  392(4):995--1003, 00 2013.

\bibitem{Daley1965}
DJ~Daley and David~G Kendall.
\newblock Stochastic rumours.
\newblock {\em IMA Journal of Applied Mathematics}, 1(1):42--55, 1965.

\bibitem{Maki1973}
D.P. Maki and M.~Thompson.
\newblock {\em Mathematical models and applications: with emphasis on the
  social, life, and management sciences}.
\newblock Prentice-Hall, 1973.

\bibitem{Nekovee2007}
M.~Nekovee, Y.~Moreno, G.~Bianconi, and M.~Marsili.
\newblock Theory of rumour spreading in complex social networks.
\newblock {\em Physica A: Statistical Mechanics and its Applications},
  374(1):457 -- 470, 2007.

\bibitem{Grabowicz2012}
Przemyslaw~A Grabowicz, Jos{\'e}J Ramasco, Esteban Moro, Josep~M Pujol, and
  Victor~M Eguiluz.
\newblock Social features of online networks: the strength of intermediary ties
  in online social media.
\newblock {\em PLoS ONE}, 7(1):e29358, 00 2012.

\bibitem{Rodrigues2011}
Tiago Rodrigues, Fabr{\'\i}cio Benevenuto, Meeyoung Cha, Krishna Gummadi, and
  Virg{\'\i}lio Almeida.
\newblock On word-of-mouth based discovery of the web.
\newblock In {\em Proceedings of the 2011 ACM SIGCOMM conference on Internet
  measurement conference}, pages 381--396. ACM, 2011.

\bibitem{Hernandez-Ceron2013}
Nancy Hernandez-Ceron, Zhilan Feng, and Carlos Castillo-Chavez.
\newblock Discrete epidemic models with arbitrary stage distributions and
  applications to disease control.
\newblock pages 1--31, 2013.

\bibitem{allen2007}
Linda J.~S Allen.
\newblock {\em An introduction to mathematical biology}.
\newblock Pearson/Prentice Hall, Upper Saddle River, NJ, 2007.

\bibitem{richard1996first}
Richard~A. HOLMGREN.
\newblock {\em A First Course in Discrete Dynamical Systens}.
\newblock Springer, 1996.

\bibitem{Allen2008}
Linda~JS Allen and P~van~den Driessche.
\newblock The basic reproduction number in some discrete-time epidemic models.
\newblock {\em Journal of Difference Equations and Applications},
  14(10-11):1127--1147, 2008.

\bibitem{Budak2011}
Ceren Budak, Divyakant Agrawal, and Amr El~Abbadi.
\newblock Limiting the spread of misinformation in social networks.
\newblock In {\em Proceedings of the 20th international conference on World
  wide web}, pages 665--674. ACM, 2011.

\bibitem{Wang2012}
Feng Wang, Haiyan Wang, and Kuai Xu.
\newblock Diffusive logistic model towards predicting information diffusion in
  online social networks.
\newblock In {\em Distributed Computing Systems Workshops (ICDCSW), 2012 32nd
  International Conference on}, pages 133--139. IEEE, 2012.

\bibitem{Han2012}
Xiaoting Han and Li~Niu.
\newblock Word of mouth propagation in online social networks.
\newblock {\em Journal of Networks}, 7(10):1670--1676, 2012.

\bibitem{Tripathy2010}
Rudra~M Tripathy, Amitabha Bagchi, and Sameep Mehta.
\newblock A study of rumor control strategies on social networks.
\newblock page 1817, New York, New York, USA, 00 2010. SIGIR, ACM Special
  Interest Group on Information Retrieval.

\end{thebibliography}
\end{document}